\documentclass[a4paper,11pt,twoside]{article}
\usepackage[utf8]{inputenc}
\usepackage{amsfonts}
\usepackage{amsmath,amssymb,graphicx,epsfig}
\usepackage{enumerate,amsthm,epstopdf}
\usepackage{bbm}
\usepackage{dsfont}
\usepackage{algpseudocode}
\usepackage{algorithmicx}
\usepackage{algorithm}
\usepackage{listings}
\usepackage{natbib}
\usepackage{hyperref}
\hypersetup{colorlinks,linkcolor={blue},citecolor={bred},urlcolor={blue}}
\setcitestyle{round,comma}
\usepackage{cleveref}
\usepackage{xcolor}

\definecolor{bred}{rgb}{0.8,0,0}
\usepackage[a4paper,hscale=0.7,vscale=0.75,centering]{geometry}
\usepackage{fullpage}
\usepackage{authblk}
\usepackage[utf8]{inputenc}
\usepackage[T1]{fontenc}
\usepackage[english]{babel}
\usepackage{charter}

\usepackage[makeroom]{cancel}
\usepackage{authblk}
\usepackage{verbatim}

\DeclareMathOperator{\var}{var}

\def\cF{{\mathcal F}}

\def\cP{{\mathcal P}}
\def\sX{{\mathsf X}}

\def\bR{{\mathbb R}}

\def\bE{{\mathbb E}}

\def\bN{{\mathbb N}}

\def\f0{{\mathbf 0}}

\def\md{{\mathrm{d}}}

\newcommand{\cblue}{\textcolor{blue}}

\newtheorem{thm}{Theorem}

\newtheorem{assumption}{Assumption}

\newtheorem{lem}{Lemma}

\newtheorem{prop}{Proposition}

\theoremstyle{definition}

\newtheorem{rem}{Remark}

\title{Convergence rates for optimised adaptive importance samplers}

\author[$\cblue{\star}$,$\cblue{\dagger}$]{ \"Omer Deniz Akyildiz}
\author[$\cblue{\ddagger}$]{Joaqu\'in M\'iguez}
\affil[$\cblue{\star}$]{The Alan Turing Institute, London, UK.}
\affil[$\cblue{\dagger}$]{University of Warwick, Coventry, UK.}
\affil[$\cblue{\ddagger}$]{Universidad Carlos III de Madrid, Leganes, Spain.}
\affil[ ]{{\textcolor{bred}{\footnotesize \texttt{omer.akyildiz@warwick.ac.uk, joaquin.miguez@uc3m.es}}}}
  
\begin{document}
\maketitle
\begin{abstract}
Adaptive importance samplers are adaptive Monte Carlo algorithms to estimate expectations with respect to some target distribution which \textit{adapt} themselves to obtain better estimators over a sequence of iterations. Although it is straightforward to show that they have the same $\mathcal{O}(1/\sqrt{N})$ convergence rate as standard importance samplers, where $N$ is the number of Monte Carlo samples, the behaviour of adaptive importance samplers over the number of iterations has been left relatively unexplored. In this work, we investigate an adaptation strategy based on convex optimisation which leads to a class of adaptive importance samplers termed \textit{optimised adaptive importance samplers} (OAIS). These samplers rely on the iterative minimisation of the $\chi^2$-divergence between an exponential-family proposal and the target. The analysed algorithms are closely related to the class of adaptive importance samplers which minimise the variance of the weight function. We first prove non-asymptotic error bounds for the mean squared errors (MSEs) of these algorithms, which explicitly depend on the number of iterations and the number of samples together. The non-asymptotic bounds derived in this paper imply that when the target belongs to the exponential family, the $L_2$ errors of the optimised samplers {converge to the optimal rate of $\mathcal{O}(1/\sqrt{N})$ and the rate of convergence in the number of iterations are explicitly provided. When the target does {\em not} belong to the exponential family, the rate of convergence is the same but the asymptotic $L_2$ error increases by a factor $\sqrt{\rho^\star} > 1$, where $\rho^\star - 1$ is the minimum $\chi^2$-divergence between the target and an exponential-family proposal.}
\end{abstract}

\section{Introduction}\label{secIntro}
The class of adaptive importance sampling (AIS) methods is a key Monte Carlo methodology for estimating integrals that cannot be obtained in closed form \citep{robert2004monte}. This problem arises in many settings, such as Bayesian signal processing and machine learning \citep{bugallo2015adaptive, bugallo2017adaptive} or optimal control, \citep{kappen2016adaptive} where the quantities of interest are usually defined as intractable expectations. Adaptive importance samplers are versions of classical importance samplers (IS) which iteratively improve the proposals to generate samples better suited to the estimation problem at hand. Its variants include, for example, \textit{population Monte Carlo} methods \citep{cappe2004population} and adaptive mixture importance sampling \citep{cappe2008adaptive}. Since there has been a surge of papers on the topic of AIS recently, a comprehensive review is beyond the scope of this article; see e.g. \cite{bugallo2017adaptive} for a recent review.

Due to the popularity of the adaptive importance samplers, their theoretical performance has also received attention in the past few years. The same as conventional IS methods, AIS schemes enjoy the classical $\mathcal{O}(1/\sqrt{N})$ convergence rate of the $L_2$ error, where $N$ is the number of Monte Carlo samples used in the approximations, see e.g. \cite{robert2004monte} and \cite{agapiou2017importance}. However, since an adaptation is performed over the iterations and the goal of this adaptation is to improve the proposal quality, an insightful convergence result would provide a bound which explicitly depends on the number of iterations, $t$, (which sometimes we refer to as \textit{time}) and the number of samples, $N$. Although there are convergence results of adaptive methods (see \cite{douc2007convergence} for a convergence theory for population Monte Carlo based on minimizing Kullback-Leibler divergence), none of the available results yields an explicit bound of the error in terms of the number of iterations and the number of particles at the same time.

One difficulty of proving such a result for adaptive mixture samplers is that the adaptive mixtures form an interacting particle system and it is unclear what kind of adaptation they perform or whether the adapted proposals actually get closer to the target for some metric. An alternative to adaptation using mixtures is the idea of minimizing a cost function in order to adapt the proposal. This idea has been popular in the literature, in particular, minimizing the variance of the weight function has received significant attention, see, e.g., \citet{arouna2004adaptative,arouna2004robbins, kawai2008adaptive, lapeyre2011framework, ryu2014adaptive, kawai2017acceleration, kawai2018optimizing}. Relevant to us, in particular, is the work of \citet{ryu2014adaptive}, who have have proposed an algorithm called Convex Adaptive Monte Carlo (Convex AdaMC). This scheme is based on minimizing the variance of the IS estimator, which is a quantity related to the $\chi^2$ divergence between the target and the proposal. \citet{ryu2014adaptive} have shown that the variance of the IS estimator is a convex function of the parameters of the proposal when the latter is chosen within the exponential family. Based on this observation, \citet{ryu2014adaptive} have formulated Convex AdaMC, which draws one sample at each iteration and construct the IS estimator, which requires access to the normalised target. They proved a central limit theorem (CLT) for the resulting sampler. The idea has been further extended for self-normalised importance samplers by \citet{ryu2016convex}, who considered minimising the $\alpha$-divergence between the target and an exponential family. Similarly, \citet{ryu2016convex} proved a CLT for the resulting sampler. Similar ideas were also considered by \citet{kawai2017acceleration, kawai2018optimizing}, who also aimed at minimizing the variance expression. Similarly, \citet{kawai2018optimizing} showed that the variance of the weight function is convex when the proposal family is suitably chosen and provided general conditions for such proposals. \citet{kawai2018optimizing} has also developed an adaptation technique based on the stochastic approximation, which is similar to the scheme we analyse in this paper. {There have been other results also considering $\chi^2$ divergence and relating it to the necessary sample size of the IS methods, see, e.g., \citet{sanz2018importance}. Following the approach of \citet{chatterjee2018sample}, \citet{sanz2018importance} considers and ties the necessary sample size to $\chi^2$-divergence, in particular, shows that the necessary sample size grows with $\chi^2$-divergence, hence implying that minimizing it can lead to more efficient importance sampling procedures.}

In this work, we develop and analyse a family of adaptive importance samplers, coined \textit{optimised adaptive importance samplers} (OAIS), which relies on a particular adaptation strategy based on convex optimisation. We adapt the proposal with respect to a quantity (essentially the $\chi^2$-divergence between the target and the proposal) that also happens to be the constant in the error bounds of the IS (see, e.g., \citep{agapiou2017importance}). Assuming that proposal distributions belong to the exponential family, we recast the adaptation of the proposal as a convex optimisation problem and then develop a procedure which essentially optimises the $L_2$ error bound of the algorithm. By using results from convex optimisation, we obtain error rates depending on the number of iterations, denoted as $t$, and the number of Monte Carlo samples, denoted as $N$, together. In this way, we explicitly display the trade-off between these two essential quantities. To the best of our knowledge, none of the papers on the topic provides convergence rates depending explicitly on the number of iterations and the number of particles together, as we do herein.

The paper is organised as follows. In Sec.~\ref{sec:AISintro}, we introduce the problem definition, the IS and the AIS algorithms. In Sec.~\ref{sec:theAlg}, we introduce the OAIS algorithms. In Sec.~\ref{sec:analysis}, we provide the theoretical results regarding optimised AIS and show its convergence using results from convex optimisation. Finally, we make some concluding remarks in Sec.~\ref{sec:conc}.

\subsection*{Notation}

For $L\in\bN$, we use the shorthand $[L] = \{1,\ldots,L\}$. We denote the state space as $\sX$ and assume $\sX \subseteq \bR^{d_x}$, $d_x \ge 1$. The space of bounded real-valued functions and the set of probability measures on space $\sX$ are denoted as $B(\sX)$ and $\cP(\sX)$, respectively. Given $\varphi\in B(\sX)$ and $\pi\in\cP(\sX)$, the expectation of $\varphi$ with respect to (w.r.t.) $\pi$ is written as $(\varphi,\pi) = \int \varphi(x) \pi(\mbox{d}x)$ or $\bE_\pi[\varphi(X)]$. The variance of $\varphi$ w.r.t. $\pi$ is defined as $\var_\pi(\varphi) = (\varphi^2,\pi) - (\varphi,\pi)^2$. If $\varphi\in B(\sX)$, then $\|\varphi\|_\infty = \sup_{x\in\sX} |\varphi(x)| < \infty$. The unnormalised density associated to $\pi$ is denoted with $\Pi(x)$. We denote the proposal as $q_\theta \in \cP(\sX)$, with an explicit dependence on the parameter $\theta\in\Theta$. The parameter space is assumed to be a subset of $d_\theta$-dimensional Euclidean space, i.e., $\Theta \subseteq \bR^{d_\theta}$. 

Whenever necessary we denote both the probability measures, $\pi$ and $q_\theta$, and their densities with the same notation. To be specific, we assume that both $\pi(\mbox{d}x)$ and $q_\theta(\mbox{d}x)$ are absolutely continuous with respect to the Lebesgue measure and we denote their associated densities as $\pi(x)$ and $q_\theta(x)$. The use of either the measure or the density will be clear from both the argument (sets or points, respectively) and the context.

\section{Background}\label{sec:AISintro}

In this section, we review importance and adaptive importance samplers.

\subsection{Importance sampling}

Consider a target density $\pi \in \cP(\sX)$ and a bounded function $\varphi \in B(\sX)$. Often, the main interest is to compute an integral of the form
\begin{align}\label{eq:ProbDefn}
(\varphi,\pi) = \int_\sX \varphi(x) \pi(x) \mbox{d}x.
\end{align}
While perfect Monte Carlo can be used to estimate this expectation when it is possible to sample exactly from $\pi(x)$, this is in general not tractable. Hereafter, we consider the cases when the target can be evaluated exactly and up to a normalising constant, respectively.

Importance sampling (IS) uses a proposal distribution which is easy to sample and evaluate. The method consists in weighting these samples, in order to correct the discrepancy between the target and the proposal, and finally constructing an estimator of the integral. To be precise, let $q_\theta\in\cP(\sX)$ be the proposal which is parameterized by the vector $\theta\in\Theta$. The unnormalised target density is denoted as $\Pi:\sX \to \bR_+$. Therefore, we have
\begin{align*}
\pi(x) = \frac{\Pi(x)}{Z_\pi},
\end{align*}
where $Z_\pi :=\int_\sX \Pi(x) \md x < \infty$. Next, we define functions $w_\theta, W_\theta:\sX \times \Theta \to \bR_+$ as
\begin{align*}
w_\theta(x) = \frac{\pi(x)}{q_\theta(x)} \quad \textnormal{and} \quad W_\theta(x) = \frac{\Pi(x)}{q_\theta(x)},
\end{align*}
respectively. For a chosen proposal $q_\theta$, the IS proceeds as follows. First, a set of independent and identically distributed (iid) samples $\{x^{(i)}\}_{i=1}^N$ is generated from $q_\theta$. When $\pi(x)$ can be evaluated, one constructs the empirical approximation of the probability measure $\pi$, denoted $\pi_\theta^N$, as
\begin{align*}
\pi_\theta^N(\mbox{d}x) = \frac{1}{N} \sum_{i=1}^N w_\theta(x^{(i)})\delta_{x^{(i)}}(\mbox{d}x),
\end{align*}
where $\delta_{x'}(\mbox{d}x)$ denotes the Dirac delta measure that places unit probability mass at $x=x'$. For this case, the IS estimate of the integral in \eqref{eq:ProbDefn} can be given as
\begin{align}\label{eq:ISestimate}
(\varphi,\pi^N_\theta) = \frac{1}{N} \sum_{i=1}^N w_\theta(x^{(i)}) \varphi(x^{(i)}).
\end{align}
However, in most practical cases, the target density $\pi(x)$ can only be evaluated up to an unknown normalizing proportionality constant (i.e., we can evaluate $\Pi(x)$ but not $Z_\pi$). In this case, we construct the empirical measure $\pi^N_\theta$ as
\begin{align*}
\pi_\theta^N(\mbox{d}x) = \sum_{i=1}^N \mathsf{w}_\theta^{(i)} \delta_{x^{(i)}}(\mbox{d}x),
\end{align*}
where
\begin{align*}
\mathsf{w}_\theta^{(i)} = \frac{W_\theta(x^{(i)})}{\sum_{j=1}^N W_\theta(x^{(j)})}.
\end{align*}
Finally this construction leads to the so called self-normalizing importance sampling (SNIS) estimator
\begin{align}\label{eq:SNISestimate}
(\varphi,\pi^N_\theta) = \sum_{i=1}^N \mathsf{w}_\theta^{(i)} \varphi(x^{(i)}).
\end{align}
Although the IS estimator \eqref{eq:ISestimate} is unbiased, the SNIS estimator \eqref{eq:SNISestimate} is in general biased. However, the bias and the MSE vanish with a rate $\mathcal{O}(1/N)$, therefore providing guarantees of convergence as $N\to\infty$. Crucially for us, the MSE of both estimators. {For clarity, below we present an MSE bound for the (more general) SNIS estimator \eqref{eq:SNISestimate} which is adapted from \citet{agapiou2017importance}.}
\begin{thm}\label{thm:ISfund}
Assume that $(W_\theta^2,q_\theta) < \infty$. Then for any $\varphi\in B(\sX)$, we have
\begin{align}
\bE\left[\left((\varphi,\pi) - (\varphi,\pi_{\theta}^N)\right)^2\right] \leq \frac{c_\varphi \rho(\theta)}{{N}},
\label{eqThm1-1}
\end{align}
where $c_\varphi = 4\|\varphi\|_\infty^2$ and the function $\rho:\Theta \to [\rho^\star,\infty)$ is defined as
\begin{align}
\rho(\theta) = \bE_{q_\theta}\left[\frac{\pi^2(X)}{q^2_\theta(X)}\right],
\label{eqThm1-2}
\end{align}
where $\rho^\star := \inf_{\theta\in\Theta} \rho(\theta) \geq 1$.
\end{thm}
\begin{proof}
See Appendix \ref{app:proofIS} for a self-contained proof.
\end{proof}
\begin{rem} For the IS estimator \eqref{eq:ISestimate}, this bound can be improved so that $c_\varphi~=~\|\varphi\|_\infty^2$. However, this improvement does not effect our results in this paper, hence we present a single bound of the form in \eqref{eqThm1-1} for the estimators \eqref{eq:ISestimate} and \eqref{eq:SNISestimate} for conciseness. $\square$
\end{rem}
\begin{rem}\label{rem:relationToChi} As pointed out by \cite{agapiou2017importance}, the function $\rho$ is essentially the $\chi^2$ divergence between $\pi$ and $q_\theta$, i.e.,
\begin{align*}
\rho(\theta) := \chi^2(\pi || q_\theta) + 1.
\end{align*}
Note that $\rho(\theta)$ can also be expressed in terms of the variance of the weight function $w_\theta$, which coincides with the $\chi^2$-divergence, i.e.,
\begin{align*}
\rho(\theta) = \var_{q_\theta}(w_\theta(X)) + 1.
\end{align*}
Therefore, minimizing $\rho(\theta)$ is equivalent to minimizing $\chi^2$-divergence and the variance of the weight function $w_\theta$, i.e., $\var_{q_\theta}(w_\theta(X))$. $\square$
\end{rem}
\begin{rem} Remark~\ref{rem:relationToChi} implies that, when both $\pi$  and $q_\theta$ belong {to the same parametric family (i.e., there exists $\theta \in \Theta$ such that $\pi=q_\theta$),} one readily obtains
\begin{align*}
\rho^\star := \inf_{\theta\in\Theta} \rho(\theta) = 1. \quad \square
\end{align*}
\end{rem}
\begin{rem} For the IS estimator \eqref{eq:ISestimate}, the bound in Theorem~\ref{thm:ISfund} can be modified so that it holds for unbounded test functions $\varphi$ as well; see, e.g. \citet{ryu2014adaptive}. Therefore, a similar quantity to $\rho(\theta)$, which includes $\varphi$ whilst still retaining convexity, can be optimised for this case. Unfortunately, obtaining such a bound is not straightforward for the SNIS estimator \eqref{eq:SNISestimate} as shown by \citet{agapiou2017importance}. In order to significantly simplify the presentation, we restrict ourselves to the class of bounded test functions, i.e., we assume $\|\varphi\|_\infty < \infty$. $\square$
\end{rem}
{Finally, we present a bias result from \citet{agapiou2017importance}.
\begin{thm}\label{thm:SNISbias}
Assume that $(W_\theta^2,q_\theta) < \infty$. Then for any $\varphi\in B(\sX)$, we have
\begin{align*}
\left| \bE\left[(\varphi,\pi_{\theta}^N)\right] - (\varphi,\pi) \right| \leq \frac{\bar{c}_\varphi \rho(\theta)}{{N}},
\end{align*}
where $\bar{c}_\varphi = 12\|\varphi\|_\infty^2$ and the function $\rho:\Theta \to [\rho^\star,\infty)$ is the same as in Theorem~\ref{thm:ISfund}.
\end{thm}
\begin{proof}
See Theorem 2.1 in \citet{agapiou2017importance}.
\end{proof}}
\subsection{Parametric adaptive importance samplers}
Standard importance sampling may be inefficient in practice when the proposal is poorly calibrated with respect to the target. In particular, as implied by the error bound provided in Theorem~\ref{thm:ISfund}, the error made by the IS estimator can be high if the $\chi^2$-divergence between the target and the proposal is large. Therefore, it is more common to employ an iterative version of importance sampling, also called as \textit{adaptive importance sampling} (AIS). The AIS algorithms are importance sampling methods which aim at iteratively improving the proposal distributions. More specifically, the AIS methods specify a sequence of proposals $(q_t)_{t\geq 1}$ and perform importance sampling at each iteration. The aim is to improve the proposal so that the samples are better matched with the target, which results in less variance and more accuracy in the estimators. There are several variants, the most popular one being population Monte Carlo methods \citep{cappe2004population} which uses previous samples in the proposal.

\begin{algorithm}[t]
\begin{algorithmic}[1]
\caption{Parametric AIS}\label{alg:ParametricAIS}
\State Choose a parametric proposal $q_{\theta}$ with initial parameter $\theta=\theta_0$.
\For{$t\geq 1$}
\State Adapt the proposal,
\begin{align*}
\theta_t = \mathcal{T}_t(\theta_{t-1}),
\end{align*}
\State Sample,
\begin{align*}
x_t^{(i)} \sim q_{\theta_t}, \quad \textnormal{for } i = 1,\ldots,N,
\end{align*}
\State Compute weights,
\begin{align*}
\mathsf{w}_{\theta_t}^{(i)} = \frac{W_{\theta_t}(x_t^{(i)})}{\sum_{i=1}^N W_{\theta_t}(x_t^{(i)})}, \quad \textnormal{where} \quad W_{\theta_t}^{(i)} = \frac{\Pi(x_t^{(i)})}{q_{\theta_t}(x^{(i)})}.
\end{align*}
\State Report the point-mass probability measure
\begin{align*}
{\pi}_{\theta_t}^N(\md x) = \sum_{i=1}^N \mathsf{w}_{\theta_t}^{(i)} \delta_{x_t^{(i)}}(\md x),
\end{align*}
and the estimator
\begin{align*}
(\varphi,{\pi}_{\theta_t}^N) = \sum_{i=1}^N \mathsf{w}_{\theta_t}^{(i)} \varphi(x_t^{(i)}).
\end{align*}
\EndFor
\end{algorithmic}
\end{algorithm}

In this section, we review one particular AIS, which we refer to as \textit{parametric AIS}. In this variant, the proposal distribution is a parametric distribution, denoted $q_\theta$. Over time, this parameter $\theta$ is updated (or \textit{optimised}) with respect to a predefined criterion resulting in a sequence $(\theta_t)_{t\geq 1}$. This yields a sequence of proposal distributions denoted as $(q_{\theta_t})_{t\geq 1}$.

One iteration of the algorithm goes as follows. Assume at time $t-1$ we are given a proposal distribution $q_{\theta_{t-1}}$. At time $t$, we first update the parameter of this proposal,
\begin{align*}
\theta_t = \mathcal{T}_t(\theta_{t-1}),
\end{align*}
where $\{\mathcal{T}_t:\Theta \to \Theta, t\geq 1\}$, is a sequence of (deterministic or stochastic) maps, e.g., gradient mappings, constructed so that they minimise a certain cost function. Then, in the same way we have done in conventional IS, we sample
\begin{align*}
x_t^{(i)} \sim q_{\theta_t}(\md x), \quad \textnormal{for } i = 1,\ldots,N,
\end{align*}
compute weights
\begin{align*}
\mathsf{w}_{\theta_t}^{(i)} = \frac{W_{\theta_t}(x_t^{(i)})}{\sum_{i=1}^N W_{\theta_t}(x_t^{(i)})},
\end{align*}
and finally construct the empirical measure
\begin{align*}
\pi_{\theta_t}^N(\md x) = \sum_{i=1}^N \mathsf{w}_{\theta_t}^{(i)} \delta_{x_t^{(i)}}(\md x).
\end{align*}
The estimator of the integral \eqref{eq:ProbDefn} is then computed as in Eq. \eqref{eq:SNISestimate}. 

The full procedure of the parametric AIS method is summarized in Algorithm~\ref{alg:ParametricAIS}. Since this is a valid IS scheme, this algorithm enjoys the same guarantee provided in Theorem~\ref{thm:ISfund}. In particular, we have the following theorem.
\begin{thm}\label{thm:ISfundAIS}
Assume that, given a sequence of proposals $(q_{\theta_t})_{t\geq 1} \in \cP(\sX)$, we have $(W_{\theta_t}^2,q_{\theta_t}) < \infty$ for every $t$. Then for any $\varphi\in B(\sX)$, we have
\begin{align*}
\bE\left[\left|(\varphi,\pi) - (\varphi,\pi_{\theta_t}^N)\right|^2\right] \leq \frac{c_\varphi \rho(\theta_t)}{{N}},
\end{align*}
where $c_\varphi = 4 \|\varphi\|_\infty^2$ and the function $\rho(\theta_t):\Theta \to [\rho^\star,\infty)$ is defined as in Eq. \eqref{eqThm1-2}.
\end{thm}
\begin{proof}
The proof is identical to the proof of Theorem~\ref{thm:ISfund}. We have just re-stated the result to introduce the iteration index $t$.
\end{proof}
However, this theorem does not give an insight of what happens as the number of iterations increases, i.e., when $t\to\infty$, with the bound. Ideally, the adaptation of the AIS should improve this bound with time. In other words, in the ideal case, the error should decrease as $t$ grows. Fortunately, Theorem~\ref{thm:ISfundAIS} suggests that the maps $\mathcal{T}_t:\Theta\to\Theta$ can be chosen so that the function $\rho$ is minimised over time. More specifically, the sequence $(\theta_t)_{t\geq 1}$ can be chosen so that it leads to a decreasing sequence (at least in expectation) $(\rho(\theta_t))_{t\geq 1}$. In the following sections, we will summarize the deterministic and stochastic strategies to achieve this aim.
\begin{rem}\label{remR} We define the unnormalised version of $\rho(\theta)$ and denote it as $R(\theta)$. It is characterised as follows
\begin{align*}
\rho(\theta) = \frac{R(\theta)}{Z_\pi^2} \quad \textnormal{where} \quad Z_\pi = \int_\sX \Pi(x) \md x < \infty.
\end{align*}
Hence, $R(\theta)$ can also be expressed as
\begin{align}\label{eq:Rtheta}
R(\theta) = \bE_{q_{\theta}} \left[\frac{\Pi^2(X)}{q_{\theta}^2(X)}\right].
\end{align} $\square$
\end{rem}

\subsection{AIS with exponential family proposals}\label{sec:expFamily}

Following \cite{ryu2014adaptive}, we note that when $q_\theta$ is chosen as an exponential family density, the function $\rho(\theta)$ is convex. In particular, we define
\begin{align}\label{eq:PropDefineExp}
q_\theta(x) = \exp(\theta^\top T(x) - A(\theta)) h(x),
\end{align}
where $A: \bR^{d_\theta}\to\bR \cup \{\infty\}$ is the log of the normalization constant, i.e.,
\begin{align*}
A(\theta) = \log \int \exp(\theta^\top T(x)) h(x) \mbox{d}x,
\end{align*}
while $T:\bR^{d_x}\to\bR^{d_\theta}$ and $h:\bR^{d_x}\to\bR_+$. Then we have the following lemma adapted from \cite{ryu2014adaptive}.
\begin{lem}\label{prop:rhoconvex} Let $q_\theta$ be chosen as in \eqref{eq:PropDefineExp}. Then $\rho:\Theta \to [\rho^\star,\infty)$ is convex, i.e., for any $\theta_1,\theta_2\in\Theta$ and $\lambda \in [0,1]$, the following inequality holds
\begin{align*}
\rho(\lambda\theta_1 + (1-\lambda) \theta_2) \leq \lambda \rho(\theta_1) + (1-\lambda) \rho(\theta_2).
\end{align*}
\end{lem}
\begin{proof}
See Appendix \ref{app:proofLemma1} for a self-contained proof.
\end{proof}
Lemma~\ref{prop:rhoconvex} shows that $\rho$ is a convex function, therefore, optimising it could give us provably convergent algorithms (as $t$ increases). Next lemma, borrowed from \citet{ryu2014adaptive}, shows that $\rho$ is differentiable and its gradient can indeed be computed as an expectation.
{\begin{lem}\label{lem:GradientRho} The gradient $\nabla\rho(\theta)$ can be written as
\begin{align}\label{eq:gradRho}
\nabla \rho(\theta) = \bE_{q_\theta} \left[(\nabla A(\theta) - T(X)) \frac{\pi^2(X)}{q_\theta^2(X)}\right].
\end{align}
\end{lem}}
\begin{proof}
The proof is straightforward since $q_\theta$ is from an exponential family and $A(\theta)$ is differentiable.
\end{proof}
\begin{rem} Note that Eqs.~\eqref{eq:Rtheta} and \eqref{eq:gradRho} together imply that
\begin{align}\label{eq:gradR}
\nabla R(\theta) = \bE_{q_\theta} \left[(\nabla A(\theta) - T(X)) \frac{\Pi^2(X)}{q_\theta^2(X)}\right].
\end{align}
We also note (see Remark~\ref{remR}) that
\begin{align}\label{eq:RelGrads}
\nabla R(\theta) = Z_\pi^2 \nabla \rho(\theta).
\end{align}
$\square$
\end{rem}
In the following sections, we assume that $\rho(\theta)$ is a convex function. Thus Lemma~\ref{prop:rhoconvex} constitutes an important motivation for our approach. We leave general proposals which lead to nonconvex $\rho(\theta)$ for future work.

\section{Algorithms}\label{sec:theAlg}

In this section, we describe adaptation strategies based on minimizing $\rho(\theta)$. In particular, we design maps $\mathcal{T}_t:\Theta\to\Theta$, for $t\geq 1$, for scenarios where
\begin{itemize}
\setlength{\itemindent}{2em}
\item[(i)] the gradient of $\rho(\theta)$ can be exactly computed,
\item[(ii)] an unbiased estimate of the gradient of $\rho(\theta)$ can be obtained, and
\item[(iii)] an unbiased estimate of the gradient of $R(\theta)$ can be obtained.
\end{itemize}
Scenario (i) is unrealistic in practice but gives us a guideline in order to further develop the idea. {In particular, the error bounds for the more complicated cases follow the same structure as this case. Therefore, the results obtained in case (i) provide a good qualitative understanding of the results introduced later.} Scenario (ii) can be realized in cases where it is possible to evaluate $\pi(x)$, in which case the IS leads to unbiased estimators. Scenario (iii) is what a practitioner would most often encounter: the target can only be evaluated up to the normalizing constant, i.e., $\Pi(x)$ can be evaluated but $\pi(x)$ cannot.

{We finally remark that, for the cases where we assume a stochastic gradient can be obtained for $\rho$ and $R$ (namely, the case (ii) and the case (iii) respectively), we consider two possible algorithms to perform adaptation. The first method is a \textit{vanilla} SGD algorithm \citep{bottou2016optimization} and the second method is a SGD scheme with iterate averaging \citep{schmidt2017minimizing}. While vanilla SGD is easier to implement and algorithmically related to population-based Monte Carlo methods, iterate averaged SGD results in a better theoretical bound and it has some desirable variance reduction properties.}

\subsection{Exact gradient OAIS}

We first introduce the OAIS scheme where we assume that the exact gradients of $\rho(\theta)$ are available. Since $\rho$ is defined as an expectation (an integral), this assumption is unrealistic. However, the results we can prove for this procedure shed light onto the results that will be proved for practical scenarios in the following sections.

In particular, in this scheme, given $\theta_{t-1}$, we specify $\mathcal{T}_t$ as
\begin{align}\label{eq:exactOAIS}
\theta_t = \mathcal{T}_t(\theta_{t-1}) = \mathsf{Proj}_\Theta(\theta_{t-1} - \gamma \nabla \rho(\theta_{t-1})),
\end{align}
where $\gamma > 0$ is the step-size parameter of the map and $\mathsf{Proj}_\Theta$ denotes projection onto the compact parameter space $\Theta$. This is a classical gradient descent scheme on $\rho(\theta)$. In Section \ref{ssErrorsExactGrad}, we provide non-asymptotic results for this scheme. However, as we have noted, this idea does not lead to a practical scheme and cannot be used in most cases in practice as the gradients of $\rho$ in exact form are rarely available.
\begin{rem} {We use a projection operator in Eq.~\eqref{eq:exactOAIS} because we assume throughout the analysis in Section~\ref{sec:analysis} that the parameter space $\Theta$ is compact.}
$\square$
\end{rem}

{\subsection{Stochastic gradient OAIS}}

{Although it has a nice and simple form, exact-gradient OAIS is often intractable as, in most practical cases, the gradient can only be estimated. In this section, we first look at the case where $\pi(x)$  can be evaluated, which means that an unbiased estimate of $\nabla \rho(\theta)$ can be obtained. Then we consider the general case, where one can only evaluate $\Pi(x)$ and can obtain an unbiased estimate of $\nabla R(\theta)$.}

{In the following subsections, we consider an algorithm where the gradient is estimated using samples which can also be used to construct importance sampling estimators. The procedure is outlined in Algorithm \ref{alg:SGDAIS} for the case in which only $\Pi(x)$ can be evaluated and $\nabla R(\theta)$ is estimated.}

{\subsubsection{Normalised case}}

{If we assume that the density $\pi(x)$ can be evaluated exactly, then the algorithm can be described as follows. Given $(\theta_k)_{1\leq k\leq t-1}$, at iteration $t$ we compute the next parameter iterate as
\begin{align*}
\theta_t = \mathsf{Proj}_{\Theta}(\theta_{t-1} - \gamma_t g_t),
\end{align*}
where $g_t$ is an unbiased estimator of $\nabla \rho(\theta_{t-1})$. We note that, due to the analytical form of $\nabla \rho$ (see Eq. \eqref{eq:gradRho}), the samples and weights generated at iteration $t-1$, i.e., $\left\{ x_{t-1}^{(i)}, w_{\theta_{t-1}}(x_{t-1}^{(i)}) \right\}_{i=1}^N$ can be reused to estimate the gradient. This makes an algorithmic connection to the population Monte Carlo methods where previous samples and weights are used to adapt the proposal \citep{cappe2004population}.}

{Given the updated parameter $\theta_t$, the algorithm first samples from the updated proposal $x_t^{(i)} \sim q_{\theta_t}$, $i=1, \ldots, N$, and then proceeds to construct the IS estimator as in \eqref{eq:ISestimate}. Namely, 
\begin{align}
(\varphi,\pi^N_{{\theta}_t}) = \frac{1}{N} \sum_{i=1}^N w_{{\theta}_t}({x}_t^{(i)}) \varphi({x}_t^{(i)}).
\label{eqEstimator_1}
\end{align}}

{\begin{algorithm}[tb!]
\begin{algorithmic}[1]
\caption{Stochastic gradient OAIS}\label{alg:vanillaSGDAIS}
\State Choose a parametric proposal $q_{\theta}$ with initial parameter $\theta=\theta_0$.
\For{$t\geq 1$}
\State Update the proposal parameter,
\begin{align*}
\theta_t = \mathsf{Proj}_\Theta(\theta_{t-1} - \gamma_t \tilde{g}_t)
\end{align*}
where $\tilde{g}_t$ is computed by approximating the expectation in Eq. \eqref{eq:gradR} using the samples $x_{t-1}^{(i)}$ and weights $\mathsf{w}_{\theta_{t-1}}^{(i)} = \Pi( x_{t-1}^{(i)} ) q_{\theta_{t-1}}(x_{t-1}^{(i)})^{-1}$, $i=1, ..., N$.
\State Sample,
\begin{align*}
{x}_t^{(i)} \sim q_{{\theta}_t}, \quad \textnormal{for } i = 1,\ldots,N,
\end{align*}
\State Compute weights,
\begin{align*}
\mathsf{w}_{{\theta}_t}^{(i)} = \frac{W_{{\theta}_t}({x}_t^{(i)})}{\sum_{i=1}^N W_{{\theta}_t}({x}_t^{(i)})}.
\end{align*}
\State Report,
\begin{align*}
{\pi}_{{\theta}_t}^N(\md x) = \sum_{i=1}^N \mathsf{w}_{{\theta}_t}^{(i)} \delta_{{x}_t^{(i)}}(\md x),
\end{align*}
and
\begin{align*}
(\varphi,{\pi}_{{\theta}_t}^N) = \sum_{i=1}^N \mathsf{w}_{{\theta}_t}^{(i)} \varphi({x}_t^{(i)}).
\end{align*}
\EndFor
\end{algorithmic}
\end{algorithm}}

{\subsubsection{Self-normalised case}}

{For the general case, where we can only evaluate $\Pi(x)$, the algorithm proceeds similarly. Given $(\theta_k)_{1\leq k\leq t-1}$, the method proceeds by first updating the parameter
\begin{align*}
\theta_t = \mathsf{Proj}_{\Theta}(\theta_{t-1} - \gamma_t \tilde{g}_t),
\end{align*}
where $\tilde{g}_t$ is an unbiased estimator of $\nabla R(\theta_{t-1})$. Given the updated parameter, we first sample $x_t^{(i)} \sim q_{\theta_t}$, $i=1, ..., N$, and then construct the SNIS estimate as in \eqref{eq:SNISestimate}, i.e., 
\begin{align*}
(\varphi,\pi^N_{{\theta}_t}) = \sum_{i=1}^N \mathsf{w}^{(i)}_{{\theta}_t} \varphi({x}_t^{(i)}).
\end{align*}
where
\begin{align*}
\mathsf{w}_{{\theta}_t}^{(i)} = \frac{W_{{\theta}_t}({x}^{(i)})}{\sum_{j=1}^N W_{{\theta}_t}({x}^{(j)})},
\end{align*}}

\subsection{Stochastic gradient OAIS with averaged iterates}

{Next, we describe a variant of the stochastic gradient OAIS that uses averages of the iterates generated by the SGD scheme \citep{schmidt2017minimizing} in order to compute the proposal densities, generate samples and compute weights. In Section \ref{sec:analysis} we show that the convergence rate for this method is better than the rate that can be guaranteed for Algorithm \ref{alg:vanillaSGDAIS}.}

\subsubsection{Normalised case}

We assume first that the density $\pi(x)$ can be evaluated. At the beginning of the $t$-th iteration, the algorithm has generated the sequence $(\theta_k)_{1\leq k \leq t-1}$. First, in order to perform the adaptive importance sampling steps, we set
\begin{align}\label{eq:AveragingSGD}
\bar{\theta}_t = \frac{1}{t}\sum_{k=0}^{t-1} \theta_k
\end{align}
and sample $\bar{x}_{t}^{(i)} \sim q_{\bar{\theta}_t}$ for $i = 1,\ldots,N$. Following the standard parametric AIS procedure (Algorithm~\ref{alg:ParametricAIS}), we obtain the estimate of $(\varphi,\pi)$ as,
\begin{align*}
(\varphi,\pi^N_{\bar{\theta}_t}) = \frac{1}{N} \sum_{i=1}^N w_{\bar{\theta}_t}(\bar{x}_t^{(i)}) \varphi(\bar{x}_t^{(i)}).
\end{align*}
Next, we update the parameter vector using the projected stochastic gradient step
\begin{align}\label{eq:recSgdAdaptNorm}
\theta_t = \mathcal{T}_t(\theta_{t-1}) = \mathsf{Proj}_\Theta(\theta_{t-1} - \gamma_t g_t),
\end{align}
where $g_t$ is an unbiased estimate of $\nabla\rho(\theta_{t-1})$, i.e., $\bE[g_t] = \nabla \rho(\theta_{t-1})$ and $\mathsf{Proj}_\Theta$ denotes  projection onto the set $\Theta$. Note that in order to estimate this gradient using \eqref{eq:gradRho}, we sample $x_t^{(i)} \sim q_{\theta_{t-1}}$ for $i = 1, \ldots, N$, and estimate the expectation in \eqref{eq:gradRho}. It is worth noting that the samples $\{ x_t^{(i)} \}_{i=1}^M$ are different from the samples $\{ \bar x_t^{(i)} \}_{i=1}^N$ used to estimate $(\varphi,\pi)$.

\subsubsection{Self-normalised case}

\begin{algorithm}[tb!]
\begin{algorithmic}[1]
\caption{Stochastic gradient OAIS with averaged iterates}\label{alg:SGDAIS}
\State Choose a parametric proposal $q_{\theta}$ with initial parameter $\theta = \theta_0$.
\For{$t\geq 1$}
\State Compute the average parameter vector
\begin{align*}
\bar{\theta}_t = \frac{1}{t} \sum_{k=0}^{t-1} \theta_k
\end{align*}
\State Sample,
\begin{align*}
\bar{x}_t^{(i)} \sim q_{\bar{\theta}_t}, \quad \textnormal{for } i = 1,\ldots,N,
\end{align*}
\State Compute weights,
\begin{align*}
\mathsf{w}_{\bar{\theta}_t}^{(i)} = \frac{W_{\bar{\theta}_t}(\bar{x}_t^{(i)})}{\sum_{i=1}^N W_{\bar{\theta}_t}(\bar{x}_t^{(i)})}.
\end{align*}
\State Report the point-mass probability measure
\begin{align*}
{\pi}_{\bar{\theta}_t}^N(\md x) = \sum_{i=1}^N \mathsf{w}_{\bar{\theta}_t}^{(i)} \delta_{\bar{x}_t^{(i)}}(\md x),
\end{align*}
and the estimator
\begin{align*}
(\varphi,{\pi}_{\bar{\theta}_t}^N) = \sum_{i=1}^N \mathsf{w}_{\bar{\theta}_t}^{(i)} \varphi(\bar{x}_t^{(i)}).
\end{align*}
\State Update the parameter vector,
\begin{align*}
\theta_t = \mathsf{Proj}_\Theta(\theta_{t-1} - \gamma_t \tilde{g}_t)
\end{align*}
where $\tilde g_t$ is an estimate of $\nabla R(\theta_{t-1})$ computed by approximating the expectation in Eq. \eqref{eq:gradR} using a set of iid samples ${x}_t^{(i)} \sim q_{\theta_{t-1}}$, $i=1, ..., N$.
\EndFor
\end{algorithmic}
\end{algorithm}
In general, $\pi(x)$ cannot be evaluated exactly, hence a stochastic unbiased estimate of $\nabla\rho(\theta)$ cannot be obtained. When the target can only be evaluated up to a normalisation constant, i.e., only $\Pi(x)$ can be computed, we can use the SNIS procedure as explained in Section~\ref{sec:AISintro}. Therefore, we introduce here the most general version of the stochastic method, coined \textit{stochastic gradient OAIS}, which uses the averaged iterates in \eqref{eq:AveragingSGD} to construct the proposal functions. The scheme is outlined in Algorithm \ref{alg:SGDAIS}.

To run this algorithm, given the parameter vector $\bar{\theta}_t$ in \eqref{eq:AveragingSGD}, we first generate a set of samples $\{\bar{x}_t^{(i)}\}_{i=1}^N$ from the proposal $q_{\bar{\theta}_t}$. Then the integral estimate given by the SNIS can be written as,
\begin{align*}
(\varphi,\pi^N_{\bar{\theta}_t}) =  \sum_{i=1}^N \mathsf{w}_{\bar{\theta}_t}^{(i)} \varphi(\bar{x}_t^{(i)}),
\end{align*}
where
\begin{align*}
\mathsf{w}_{\bar{\theta}_t}^{(i)} = \frac{W_{\bar{\theta}_t}(\bar{x}^{(i)})}{\sum_{j=1}^N W_{\bar{\theta}_t}(\bar{x}^{(j)})}.
\end{align*}
Finally, for the adaptation step, we obtain the unbiased estimate of the gradient $\nabla R(\theta)$ and adapt the parameter as
\begin{align}\label{eq:SgdUnnormalizedAdapt}
\theta_t = \mathsf{Proj}_\Theta(\theta_{t-1} - \gamma_t \tilde{g}_t)
\end{align}
where $\tilde{g}_t$ is an unbiased estimate of $\nabla R(\theta_{t-1})$, i.e., $\bE[\tilde{g}_t] = \nabla R(\theta_{t-1})$. Note that, as in the normalised case, this gradient is estimated by approximating the expectation in \eqref{eq:gradR} using iid samples $x_t^{(i)} \sim q_{\theta_{t-1}}$, $i = 1,\ldots,N$. These samples are different, again, from the set $\{ \bar x_t^{(i)} \}_{i=1}^N$ employed to estimate $(\varphi,\pi)$.
{
\begin{rem}
In Algorithm \ref{alg:SGDAIS} the samples $\{ \bar x_t^{(i)} \}_{i=1}^N$ drawn from the proposal distribution $q_{\bar \theta_{t-1}}(\md x)$ are \textit{not} used to compute the gradient estimator $\tilde g_t$ which, in turn, is needed to generate the next iterate $\theta_t$. Therefore, if we can afford to generate $T$ iterates, $\theta_0, \ldots, \theta_{T-1}$, with $T$ known before hand, and we are only interested in the estimator $(\varphi,\pi_{\bar \theta_T}^N)$ obtained at the last iteration (once the proposal function has been optimized) then it is be possible to skip steps 3--6 in Algorithm  \ref{alg:SGDAIS} up to time $T-1$. Only at time $t=T$, we would compute the average parameter vector $\bar \theta_T$, sample $\bar x_T^{(i)}$ from the proposal $q_{\bar \theta_T}(\md x)$ and generate the point-mass probability measure $\pi_{\bar \theta_T}^N$ and the estimator $(\varphi,\pi_{\bar \theta_T}^N)$ .
\end{rem}
}

\section{Analysis}\label{sec:analysis}

Theorem~\ref{thm:ISfund} yields an intuitive result about the performance of IS methods in terms of the divergence between the target $\pi$ and the proposal $q_\theta$. We now apply ideas from convex optimisation theory in order to minimize $\rho(\theta)$ and obtain finite-time, finite-sample convergence rates for the AIS procedures outlined in Section \ref{sec:theAlg}.

\subsection{Convergence rate with exact gradients} \label{ssErrorsExactGrad}

Let us first assume that we can compute the gradient of $\rho(\theta)$ exactly. In particular, we consider the update rule in Eq. \eqref{eq:exactOAIS}. For the sake of the analysis, we impose some regularity assumptions on the $\rho(\theta)$.

{
\begin{assumption}\label{ass:LipschitzCont}
Let $\rho(\theta)$ be a convex function with Lipschitz derivatives in the compact space $\Theta$. To be specific, $\rho$ is convex and differentiable, and there exists a constant $L<\infty$ such that
\begin{eqnarray}
\|\nabla \rho(\theta) - \nabla \rho(\theta')\|_2 &\leq& L \|\theta - \theta'\|_2 \nonumber
\end{eqnarray}
for any $\theta,\theta' \in \Theta$.
\end{assumption}
}

{
\begin{rem} Assumption \ref{ass:LipschitzCont} holds when the density $q_\theta(x)$ belongs to an exponential family (see Section~\ref{sec:expFamily}) and $\Theta$ is compact \citep{ryu2014adaptive}, even if it may not hold in general for $\theta \in \bR^{d_\theta}$. $\square$
\end{rem}
}

\begin{lem}\label{lem:GDconv} If Assumption~\ref{ass:LipschitzCont} holds and we set a step-size $\gamma \leq 1/L$, then the inequality
\begin{align}
\rho(\theta_t) - \rho^\star \leq \frac{\|\theta_0 - \theta^\star\|^2}{2\gamma t},
\label{eqLem3-1}
\end{align}
is satisfied for the sequence $\{\theta_t\}_{t\ge 0}$ generated by the recursion \eqref{eq:exactOAIS} where $\theta^\star$ is a minimum of $\rho$.
\end{lem}
\begin{proof}
See, e.g., \cite{nesterov2013introductory}.
\end{proof}

This rate in \eqref{eqLem3-1} is one of the most fundamental results in convex optimisation. Lemma \ref{lem:GDconv} enables us to prove the following result for the MSE of the AIS estimator adapted using exact gradient descent in Eq. \eqref{eqEstimator_1}.

\begin{thm}\label{thm:GD} Let Assumption~\ref{ass:LipschitzCont} hold and construct the sequence $(\theta_t)_{t\geq 1}$ using recursion \eqref{eq:exactOAIS}, where $(q_{\theta_t})_{t\geq 1}$ is the sequence of proposal distributions. Then, the inequality
\begin{align}\label{ineq:gOAISbound}
\bE\left[\left((\varphi,\pi) - (\varphi,\pi_{\theta_t}^N)\right)^2\right] &\leq \frac{c_\varphi \|\theta_0 - \theta^\star\|_2^2}{2 \gamma t {N}} +  \frac{c_\varphi \rho^\star}{{N}}
\end{align}
is satisfied, where $c_\varphi = 4 \|\varphi\|_\infty^2$, $0 < \gamma \leq 1/L$ and $L$ is the Lipschitz constant of the gradient $\nabla\rho(\theta)$ in Assumption \ref{ass:LipschitzCont}.
\end{thm}

\begin{proof}
See Appendix \ref{app:thm:GD}.
\end{proof}

\begin{rem}\label{remGDasymptote} Theorem~\ref{thm:GD} sheds light onto several facts. We first note that $\rho^\star$ in the error bound \eqref{ineq:gOAISbound} can be interpreted as an indicator of the quality of the parametric proposal. We recall that $\rho^\star = 1$ when both $\pi$ and $q_\theta$ belong to the same exponential family. For this special case, Theorem~\ref{thm:GD} implies that
\begin{align*}
\lim_{t\to\infty} \left\|(\varphi,\pi) - (\varphi,\pi_{\theta_t}^N)\right\|_2 \leq \mathcal{O}\left(\frac{1}{\sqrt{N}}\right).
\end{align*}
In other words, when the target and the proposal are both from the exponential family, this adaptation strategy is leading to an \textit{asymptotically optimal} Monte Carlo estimator (optimal meaning that we attain the same rate as a Monte Carlo estimator with $N$ iid samples from $\pi$). On the other hand, when $\pi$ and $q_\theta$ do not belong to the same family, we obtain
\begin{align*}
\lim_{t\to\infty} \left\|(\varphi,\pi) - (\varphi,\pi_{\theta_t}^N)\right\|_2 \leq \mathcal{O}\left(\sqrt{\frac{\rho^\star}{N}}\right),
\end{align*}
i.e., the $L_2$ rate is again asymptotically optimal, but the constant in the error bound is worse (bigger) by a factor $\sqrt{\rho^\star}>1$. $\square$
\end{rem}

This bound shows that as $t\to\infty$, what we are left with is essentially the minimum attainable IS error for a given parametric family $\{q_\theta\}_{\theta\in\Theta}$. Intuitively, when the proposal $q_\theta$ is from a different parametric family than $\pi$, the gradient OAIS optimises the error bound in order to obtain the best possible proposal. In particular, the MSE has two components: First an $\mathcal{O}(1/tN)$ component which can be made to vanish over time by improving the proposal and a second $\mathcal{O}(1/N)$ component which is related to $\rho^\star$. The quantity $\rho^\star$ is related to the minimum $\chi^2$-divergence between the target and proposal. This means that the discrepancy between the target and \textit{optimal} proposal (according to the $\chi^2$-divergence) can only be tackled by increasing $N$. This intuition is the same for the schemes we analyse in the next sections, although the rate with respect to the number of iterations necessarily worsens because of the uncertainty in the gradient estimators.

\begin{rem} When $\gamma = 1/L$, Theorem~\ref{thm:GD} implies that if $t = \mathcal{O}({L}/{\rho^\star})$ and $N = \mathcal{O}(\rho^\star / \varepsilon)$, for some $\varepsilon>0$, then we have
\begin{align*}
\bE\left[\left((\varphi,\pi) - (\varphi,\pi_{\theta_t}^N)\right)^2\right] &\leq \mathcal{O}(\varepsilon).
\end{align*}
We remark that once we choose the number of samples $N = \mathcal{O}(\rho^\star/\varepsilon)$, the number of iterations $t$ for adaptation is independent of $N$ and $\varepsilon$. $\square$
\end{rem}

\begin{rem} One can use different maps $\mathcal{T}_t$ for optimisation. For example, one can use Nesterov's accelerated gradient descent (which has more parameters than just a step size), in which case, one could prove (by a similar argument) the inequality \citep{nesterov2013introductory}
\begin{align*}
\bE\left[\left((\varphi,\pi) - (\varphi,\pi_{\theta_t}^N)\right)^2\right] &\leq \mathcal{O}\left(\frac{1}{t^2 {N}} +  \frac{\rho^\star}{{N}}\right).
\end{align*}
This is an improved convergence rate, going from $\mathcal{O}(1/t)$ to $\mathcal{O}(1/t^2)$ in the first term of the bound. $\square$
\end{rem}

\subsection{Convergence rate with averaged SGD iterates} \label{ssConvergence-Averaged-Iterates}

While, for the purpose of analysis, it is convenient to assume that the minimization of $\rho(\theta)$ can be done deterministically, this is rarely the case in practice. The `best' realistic case is that we can obtain an unbiased estimator of the gradient. {Hence, we address this scenario, under the assumption that the actual gradient functions $\nabla \rho$ and $\nabla R$ are bounded in $\Theta$ (i.e., $\rho(\theta)$ is Lipschitz in $\Theta$).}

{\begin{assumption}\label{ass:BoundedGradient} The gradient functions $\nabla \rho(\theta)$ and $\nabla R(\theta)$ are bounded in $\Theta$. To be specific, there exist finite constants $G_\rho$ and $G_R$ such that
\begin{eqnarray}
\sup_{\theta\in\Theta} \|\nabla \rho(\theta)\|_2 &<& G_\rho <\infty \quad \text{and} \nonumber\\
\sup_{\theta\in\Theta} \|\nabla R(\theta)\|_2 &<& G_R < \infty. \nonumber
\end{eqnarray}
\end{assumption}}

{We note that this is a mild assumption in the case of interest in this paper, where $\Theta \subset \bR^{d_\theta}$ is assumed to be compact.}

\subsubsection{Normalised target}\label{sec:NormalisedIS}

First, we assume that we can evaluate $\pi(x)$, which means that at iteration $t$, we can obtain an unbiased estimate of $\nabla \rho(\theta_{t-1})$, denoted $g_t$. We use the optimisation algorithms called \textit{stochastic gradient} methods, which use stochastic and unbiased estimates of the gradients to optimise a given cost function \citep{RobbinsMonro}. Particularly, we focus on optimised samplers using stochastic gradient descent (SGD) as an adaptation strategy.

{We start proving that the stochastic gradient estimates $(g_t)_{t\geq 0}$ have a finite mean-squared error (MSE) w.r.t. the true gradients. To prove this result, we need an additional regularity condition.}
{\begin{assumption}\label{ass:SupSupBound} 
The normalised target and proposal densities satisfy the inequality  
\begin{align*}
\sup_{\theta\in\Theta} \bE_{q_\theta}\left[ \left|\frac{\pi^2(X)}{q_\theta^2(X)} \frac{\partial \log q_\theta}{\partial \theta_j}(X) \right|^2 \right] =: D_\pi^j < \infty.
\end{align*}
for $j=1, \ldots, d_\theta$. We denote $D_\pi := \max_{j \in \{1,\ldots,d_\theta\}} D_\pi^j < \infty$.
\end{assumption}}

{
\begin{rem} \label{remSupSup}
Let us rewrite $D_\pi^j$ in Assumption \ref{ass:SupSupBound} in terms of the weight function, namely
\begin{align*}
D_\pi^j = \sup_{\theta\in\Theta} \bE_{q_\theta} \left[ \left| w_\theta^2(X) \frac{\partial \log q_\theta}{\partial \theta_j}(X) \right|^2 \right].
\end{align*}
When $q_\theta(x)$ belongs to the exponential family, we readily obtain
\begin{align*}
D_\pi^j = \sup_{\theta\in\Theta} \bE_{q_\theta} \left[ w_\theta^4(X) \left( \frac{\partial A(\theta)}{\partial \theta_i} -T_i(X) \right)^2 \right],
\end{align*}
where $T_i(X)$ is the $i$-th sufficient statistic for $q_\theta(x)$. Let us construct a bounding function for the weights of the form
$$
K(\theta) := \sup_{x \in {\sf X}} w_\theta(x).
$$
If we choose the compact space $\Theta$ in such a way that $K(\theta)$ is bounded, then we readily have
\begin{align*}
D_\pi^j &\le  \sup_{\theta\in\Theta} K^4(\theta) \bE_{q_\theta} \left[ \left( \frac{\partial A(\theta)}{\partial \theta_i} -T_i(X) \right)^2 \right] \\
&\le \| K \|_\infty^4 \text{Var}( T_i(X) ),
\end{align*}
where we have used the fact that $\frac{\partial^m A(\theta)}{\partial \theta_i} = \bE_{q_{\theta}}\left[ T_i^m(X) \right]$. Therefore, if the weights remain bounded in $\Theta$, a sufficient condition for Assumption \ref{ass:SupSupBound} to hold is that the sufficient statistics of the proposal distribution all have finite variances, i.e., $\max_{i \in \{1, \ldots, d_\theta\} } T_i(X) < \infty$.
\\ \\
{There are alternative conditions that, when satisfied, lead to Assumption \ref{ass:SupSupBound} holding true. As an example, in Appendix \ref{apRho2} we provide an alternative sufficient condition in terms of the function $\rho_2(\theta):=\bE[ w_\theta^4(X) ]$.}
\end{rem}
}

Now we show that when $g_t$ is an iid Monte Carlo estimator of $\nabla \rho$, we have the following finite-sample bound for the MSE.
{\begin{lem}\label{lem:GradientMonteCarlo}
If Assumption~\ref{ass:SupSupBound} holds, the following inequality holds,
\begin{align*}
\bE[\| g_t  - \nabla \rho(\theta_{t-1})\|_2^2] \leq \frac{d_\theta c_{\rho} D_\pi}{N},
\end{align*}
where $d_\theta$ is the parameter dimension and $c_\rho, D_\pi < \infty$ are constant w.r.t. $N$.
\end{lem}
\begin{proof}
See Appendix~\ref{app:lem:GradientMonteCarlo}.
\end{proof}}

In order to obtain convergence rates for the estimator $(\varphi,\pi_{\bar \theta_t}^N)$ we first recall a classical result for the SGD (see, e.g., \cite{bubeck2015convex}).
\begin{lem}\label{lem:SGDconv} 
{Let Assumptions \ref{ass:BoundedGradient} and \ref{ass:SupSupBound} hold, apply recursion \eqref{eq:recSgdAdaptNorm} and let $(g_t)_{t\geq 0}$ be the stochastic gradient estimates in Lemma~\ref{lem:GradientMonteCarlo}}. If we choose the step-size sequence $\gamma_k = \alpha / \sqrt{k}$, $1\leq k \leq t$, for any $\alpha > 0$, then
{\begin{align}\label{eq:SGDrate}
\bE[\rho(\bar{\theta}_t) - \rho^\star] \leq \frac{\bE\|\theta_0 - \theta^\star\|_2^2}{2\alpha\sqrt{t}} + \frac{\alpha d_\theta c_\rho D_\pi}{\sqrt{t} N} +  \frac{\alpha G^2_\rho}{\sqrt{t}},
\end{align}}
where $\bar{\theta}_t = \frac{1}{t}\sum_{k=0}^{t-1} \theta_k$.
\end{lem}

\begin{proof}
See Appendix \ref{app:lem:SGDconv} for a self-contained proof.
\end{proof}

We can now state the first core result of the paper, which is the convergence rate for the AIS algorithm using a SGD adaptation of the parameter vectors $\theta_t$.
\begin{thm}\label{thm:SGDAIS} 
Let Assumptions \ref{ass:BoundedGradient} and \ref{ass:SupSupBound} hold, let the sequence $(\theta_t)_{t\geq 1}$ be computed using \eqref{eq:recSgdAdaptNorm} and construct the averaged iterates $\bar{\theta}_t = \frac{1}{t} \sum_{k=0}^{t-1} \theta_k$. Then, the sequence of proposal distributions $(q_{\bar{\theta}_t})_{t\geq 1}$ satisfies the inequality
{\begin{align}\label{eq:rateSGDAIS}
\bE\left[\left((\varphi,\pi) - (\varphi,\pi_{\bar{\theta}_t}^N)\right)^2\right] &\leq \frac{c_1}{\sqrt{t}N} + \frac{c_2}{\sqrt{t} N^2} + \frac{c_3}{\sqrt{t} N} + \frac{c_4}{N}
\end{align}
for $t \ge 1$ and any $\varphi \in B({\sf X})$, where
\begin{align*}
c_1 &= \frac{c_\varphi \bE\|\theta_0 - \theta^\star\|_2^2}{2 \alpha}, \\ 
c_2 &= {c_\varphi c_\rho \alpha d_\theta D_\pi}, \\
c_3 &={c_\varphi \alpha G_\rho^2}, \\
c_4 &={c_\varphi \rho^\star},
\end{align*}
and $c_\varphi = 4\|\varphi\|_\infty^2$ are finite constants independent of $t$ and $N$.}
\end{thm}

\begin{proof}
See Appendix \ref{app:thm:SGDAIS}.
\end{proof}

\begin{rem} 
Note that the expectation on the left hand side of  \eqref{eq:rateSGDAIS} is taken w.r.t. the distribution of the measure-valued random variable $\pi_{\bar \theta_t}^N$. $\square$
\end{rem}

Theorem~\ref{thm:SGDAIS} can be interpreted similarly to Theorem~\ref{thm:GD}. One can see that the overall rate of the MSE bound is $\mathcal{O}\left({1}/{\sqrt{t}N} + {1}/{N}\right)$. This means that, as $t\to\infty$, we are only left with a rate that is optimal for the AIS for a given parametric proposal family. In particular, again, $\rho^\star$ is related to the minimal $\chi^2$-divergence between the target $\pi$ and the parametric proposal $q_\theta$. When the proposal and the target are from the same family, we are back to the case $\rho^\star = 1$, thus the adaptation leads to the optimal Monte Carlo rate $\mathcal{O}(1/\sqrt{N})$ for the $L_2$ error within this setting as well.

\subsubsection{Self-normalised estimators}

We have noted that it is possible to obtain an unbiased estimate of $\nabla\rho(\theta)$ when the normalised target $\pi(x)$ can be evaluated. However, if we can only evaluate the unnormalised density $\Pi(x)$ instead of $\pi(x)$ and use the self-normalized IS estimator, the estimate of $\nabla\rho(\theta)$ is no longer unbiased. We refer to Sec.~5  of \cite{tadic2017asymptotic} for stochastic optimisation with biased gradients for adaptive Monte Carlo, where the discussion revolves around minimizing the Kullback-Leibler divergence rather than the $\chi^2$-divergence. The results presented in \cite{tadic2017asymptotic}, however, are asymptotic, while herein we are interested in finite-time bounds. Due to the structure of the AIS scheme, it is possible to avoid working with biased gradient estimators. In particular, we can implement the proposed AIS schemes using unbiased estimators of $\nabla R(\theta)$ instead of biased estimators of $\nabla \rho(\theta)$. Since optimizing the unnormalised function $R(\theta)$ leads to the same minima as optimizing the normalised function $\rho(\theta)$, we can simply use $\nabla R(\theta)$ for the adaptation in the self-normalised case.

Similar to the argument in Section \ref{sec:NormalisedIS}, we first start the assumption below, which is the obvious counterpart of Assumption \ref{ass:SupSupBound}.
{\begin{assumption}\label{ass:SupSupBoundPi} 
The unnormalized target $\Pi(x)$ and the proposal densities $q_\theta(x)$ satisfy the inequalities
\begin{align*}
\sup_{\theta\in\Theta} \bE_{q_\theta} \left[ \left| \frac{\Pi^2(X)}{q_\theta^2(X)} \frac{\partial \log q_\theta}{\partial \theta_j}(X) \right|^2 \right] =: D_\Pi^j < \infty
\end{align*}
for $j=1, \ldots, d_\theta$. We denote $D_\Pi : = \frac{1}{d_\theta} \sum_{j=1}^{d_\theta} D_\Pi^j < \infty$.
\end{assumption}}

Remark \ref{remSupSup} holds directly for Assumption \ref{ass:SupSupBoundPi} as long as $Z_\pi<\infty$. {Next, we prove an MSE bound for the stochastic gradients $(\tilde{g}_t)_{t\geq 0}$ employed in recursion \eqref{eq:SgdUnnormalizedAdapt}, i.e., the unbiased stochastic estimates of $\nabla R(\theta)$.}

{
\begin{lem}\label{lem:GradientMonteCarloR}
If Assumptions \ref{ass:BoundedGradient} and \ref{ass:SupSupBoundPi} hold, the inequality
\begin{align*}
\bE[\| \tilde{g}_t  - \nabla R(\theta_{t-1})\|_2^2] \leq \frac{d_\theta c_R D_\Pi}{N},
\end{align*}
is satisfied, where $c_R, D_\Pi < \infty$ are constants w.r.t. of $N$.
\end{lem}
\begin{proof}
The proof is identical to the proof of Lemma~\ref{lem:GradientMonteCarlo}.
\end{proof}
}

We can now obtain explicit rates for the convergence of the unnormalized function $R(\bar \theta_t)$, evaluated at the averaged iterates $\bar \theta_t$. 

\begin{lem}\label{lem:SGDBiasedconv} 
If Assumptions \ref{ass:BoundedGradient} and \ref{ass:SupSupBoundPi} hold and the sequence $(\theta_t)_{t\ge 1}$ is computed via recursion \eqref{eq:SgdUnnormalizedAdapt}, with step-sizes $\gamma_k = \beta / \sqrt{k}$ for $1\leq k \leq t$ and $\beta > 0$, we obtain the inequality
{\begin{align}\label{eq:UnnormRate}
\bE[R(\bar{\theta}_t) - R ^\star] \leq \frac{\bE \|\theta_0 - \theta^\star\|_2^2}{2 \beta \sqrt{t}} + \frac{\beta d_\theta c_R D_\Pi}{\sqrt{t} N} +  \frac{\beta G^2_R}{\sqrt{t}}
\end{align}}
where $c_R,D_\Pi<\infty$ are constants w.r.t. $t$ and $N$. Relationship \ref{eq:UnnormRate} implies that
{\begin{align}\label{eq:NormRateWithNormConsts}
\bE[\rho(\bar{\theta}_t) - \rho^\star] \leq 
\frac{\bE \|\theta_0 - \theta^\star\|_2^2}{2 \beta Z_\pi^2 \sqrt{t}} + \frac{\beta d_\theta c_R D_\Pi}{Z_\pi^2 \sqrt{t} N} +  \frac{\beta G^2_R}{Z_\pi^2 \sqrt{t}}.
\end{align}}
\end{lem}
\begin{proof} The proof of the rate in \eqref{eq:UnnormRate} is identical to the proof of Lemma~\ref{lem:SGDconv}. The rate in \eqref{eq:NormRateWithNormConsts} follows by observing that $\rho(\theta) = R(\theta) / Z_\pi^2$ for every $\theta\in\Theta$.
\end{proof}

Finally, using Lemma~\ref{lem:SGDBiasedconv}, we can state our main result: an explicit error rate for the MSE of Algorithm~\ref{alg:SGDAIS} as a function of the number of iterations $t$ and the number of samples $N$.

\begin{thm}\label{thm:SGDAISUN} 
Let Assumptions \ref{ass:BoundedGradient} and \ref{ass:SupSupBoundPi} hold and let the sequence $(\theta_t)_{t\ge 1}$ be computed via recursion \eqref{eq:SgdUnnormalizedAdapt}, with step-sizes $\gamma_k = \beta / \sqrt{k}$ for $1\leq k \leq t$ and $\beta > 0$. We have the following inequality for the sequence of proposal distributions $(q_{\bar{\theta}_t})_{t\geq 1}$,
{\begin{align}\label{eq:rateUnnormAIS}
\bE\left[\left((\varphi,\pi) - (\varphi,\pi_{\bar{\theta}_t}^N)\right)^2\right] &\leq \frac{C_1}{\sqrt{t} N} + \frac{C_2}{\sqrt{t}N^2} + \frac{C_3}{\sqrt{t} N} + \frac{C_4}{N},
\end{align}
where
\begin{align*}
C_1 &= \frac{c_\varphi \bE\|\theta_0 - \theta^\star\|_2^2}{2 \beta Z_\pi^2}, \\
C_2 &= \frac{c_\varphi \beta c_R d_\theta D_\Pi}{Z_\pi^2}, \\
C_3 &= \frac{c_\varphi \beta G^2_R}{Z_\pi^2}, \\
C_4 &= {c_\varphi \rho^\star},
\end{align*}
and $c_\varphi = 4\|\varphi\|_\infty^2$ are finite constants independent of $t$ and $N$.}
\end{thm}
\begin{proof}
The proof follows from Lemma~\ref{lem:SGDBiasedconv} and mimicking the exact same steps as in the proof of Theorem~\ref{thm:SGDAIS}.
\end{proof}

\begin{rem} 
Theorem~\ref{thm:SGDAISUN}, as in Remark~\ref{remGDasymptote}, provides relevant insights regarding the performance of the stochastic gradient OAIS algorithm. In particular, for a general target $\pi$, we obtain
\begin{align*}
\lim_{t\to\infty} \left\|(\varphi,\pi) - (\varphi,\pi_{\bar{\theta}_t}^N)\right\|_2 = \mathcal{O}\left(\sqrt{\frac{\rho^\star}{N}}\right).
\end{align*}
This result shows that the $L_2$ error is asymptotically optimal. As in previous cases, if the target $\pi$ is in the exponential family, then the asymptotic convergence rate is $\mathcal{O}(1/\sqrt{N})$ as $t \to \infty$. $\square$
\end{rem}

\begin{rem} 
Theorem~\ref{thm:SGDAISUN} also yields a practical heuristic to tune the step-size and the number of particles together. Assume that $0 < \beta < 1$ and let $N = 1/\beta$ (which we assume to be an integer without loss of generality). In this case, the rate \eqref{eq:rateUnnormAIS} simplifies into
\begin{align*}
\bE\left[\left((\varphi,\pi) - (\varphi,\pi_{\bar{\theta}_t}^N)\right)^2\right] &\leq \frac{c_\varphi \bE\|\theta_0 - \theta^\star\|_2^2}{2 Z_\pi^2 \sqrt{t}} + \frac{c_\varphi \beta^3 c_R d_\theta D_\Pi}{Z_\pi^2 \sqrt{t}} + \frac{c_\varphi \beta^2 G^2_R}{Z_\pi^2\sqrt{t}} + c_\varphi \rho^\star \beta
\end{align*}
Now, if we let $t = \mathcal{O}(1/\beta^2)$ we readily obtain
\begin{align*}
\bE\left[\left((\varphi,\pi) - (\varphi,\pi_{\bar{\theta}_t}^N)\right)^2\right] &\leq \mathcal{O}(\beta).
\end{align*}
Therefore, one can control the error using the step-size of the optimisation scheme provided that other parameters of the algorithm are chosen accordingly. The same argument also holds for Theorem~\ref{thm:SGDAIS}. $\square$
\end{rem}

\begin{rem} 
{It is not straightforward to compare the rates in inequality \eqref{eq:rateUnnormAIS} (for the unnormalized target $\Pi(x)$) and inequality \eqref{eq:rateSGDAIS} (for the normalized target $\pi(x)$). Even if \eqref{eq:rateUnnormAIS} may ``look better'' by a constant factor compared to the rate in  \eqref{eq:rateSGDAIS}, this is usually not the case. Indeed, the variance of the errors in the unnormalised gradient estimators is typically higher and this reflects on the variance of the moment estimators. Another way to look at this issue is to realise that, very often, $Z_\pi << 1$, which makes the bound in \eqref{eq:rateUnnormAIS} much greater than the bound in \eqref{eq:rateSGDAIS}.}
\end{rem}

{Finally, we can adapt Theorem~\ref{thm:SNISbias} to our case, providing a convergence rate of the bias of the importance sampler given by Algorithm~\ref{alg:SGDAIS}.}

{\begin{thm}\label{thm:SGDAISbias}
Under the setting of Theorem~\ref{thm:SGDAISUN}, we have
\begin{align}\label{eq:rateUnnormAISbias}
\left| \bE\left[(\varphi,\pi_{\bar{\theta}_t}^N)\right] - (\varphi, \pi)\right| &\leq \frac{3C_1}{\sqrt{t} N} + \frac{3C_2}{\sqrt{t}N^2} + \frac{3C_3}{\sqrt{t} N} + \frac{3C_4}{N},
\end{align}
where $C_1,C_2,C_3,C_4$ are finite constants given in Theorem~\ref{thm:SGDAISUN} and independent of $t$ and $N$.
\end{thm}}
\begin{proof}
The proof follows from Theorem~\ref{thm:SNISbias} and mimicking the same proof technique used to prove Theorem~\ref{thm:SGDAISUN}.
\end{proof}

\subsection{Convergence rate with vanilla SGD}

{The arguments of Section \ref{ssConvergence-Averaged-Iterates} can be carried over to the analysis of Algorithm \ref{alg:vanillaSGDAIS}, where the proposal functions $q_{\theta_t}(x)$ are constructed using the iterates $\theta_t$ rather than the averages $\bar \theta_t$. Unfortunately, achieving the optimal $\mathcal{O}(1/\sqrt{t})$ rate for the vanilla SGD is difficult in general. The best available rate without significant restrictions on the step-size is given by \citet{shamir2013stochastic}. In particular, we can adapt \citet[Theorem~2]{shamir2013stochastic} to our setting in order to state the following lemma.
\begin{lem}\label{lem:vanillaSGDconv} 
Apply recursion \eqref{eq:SgdUnnormalizedAdapt} for the computation of the iterates $(\theta_t)_{t\ge 1}$, choose the step-sizes $\gamma_k = \beta / \sqrt{k}$ for $1\leq k \leq t$, where $\beta > 0$, and let Assumptions \ref{ass:BoundedGradient} and \ref{ass:SupSupBoundPi} hold. Then, we have the inequality
\begin{align}
\bE[R({\theta}_t) - R ^\star] \leq \left(\frac{D^2}{\beta \sqrt{t}} + \frac{\beta d_\theta c_R D_\Pi}{\sqrt{t} N} +  \frac{\beta G^2_R}{\sqrt{t}}\right) (2  + \log t),
\end{align}
where $D := \sup_{\theta,\theta' \in \Theta} \|\theta - \theta'\| < \infty$. This in turn implies that
{\begin{align}
\bE[\rho({\theta}_t) - \rho^\star] \leq 
\left(\frac{D^2}{\beta \sqrt{t}} + \frac{\beta d_\theta c_R D_\Pi}{\sqrt{t} N} +  \frac{\beta G^2_R}{\sqrt{t}}\right)\frac{(2  + \log t)}{Z_\pi^2}.
\end{align}}
\end{lem}}
\begin{proof}
It is straightforward to prove this result using \citet[Theorem~2]{shamir2013stochastic} and the proof of Lemma~\ref{lem:SGDconv}.
\end{proof}
{The obtained rate is, in general, $\mathcal{O}\left( \frac{\log t}{\sqrt{t}}\right)$. This is known to be suboptimal and it can be improved to the {information-theoretical optimal} $\mathcal{O}(1/\sqrt{t})$ rate by choosing a specific step-size scheduling, see, e.g., \citet{jain2019making}. {However, in this case, the scheduling of $(\gamma_t)_{t\geq 1}$ depends directly on the total number of iterates to be generated, in such a way that the error $\mathcal{O}(1/\sqrt{t})$ is guaranteed only for the {\em last} iterate, at the final time $t$.}}

We can extend Lemma \ref{lem:vanillaSGDconv} to obtain the following result.

{\begin{thm}\label{thm:vanillaSGDAIS} 
Apply recursion \eqref{eq:SgdUnnormalizedAdapt} for the computation of the iterates $(\theta_t)_{t\ge 1}$, choose the step-sizes $\gamma_k = \beta / \sqrt{k}$ for $1\leq k \leq t$, where $\beta > 0$, and let Assumptions \ref{ass:BoundedGradient} and \ref{ass:SupSupBoundPi} hold. If we construct the sequence of proposal distributions $(q_{{\theta}_t})_{t\geq 1}$ be the sequence of proposal distributions we obtain the following MSE bounds
\begin{align}
\bE\left[\left((\varphi,\pi) - (\varphi,\pi_{{\theta}_t}^N)\right)^2\right] &\le \left(
	\frac{C_1}{\sqrt{t} N} + \frac{C_2}{\sqrt{t}N^2} + 
	\frac{C_3}{\sqrt{t} N}
\right)(2 + \log t) + \frac{C_4}{N},
\label{eq:rateUnnormAIS-2}
\end{align}
where
\begin{align*}
C_1 &= \frac{c_\varphi D^2}{2 \beta Z_\pi^2}, \\
C_2 &= \frac{c_\varphi \beta c_R d_\theta D_\Pi}{Z_\pi^2}, \\
C_3 &= \frac{c_\varphi \beta G^2_R}{Z_\pi^2},  \\
C_4 &= {c_\varphi \rho^\star},
\end{align*}
and $c_\varphi = 4\|\varphi\|_\infty^2$ are finite constants independent of $t$ and $N$.
\end{thm}}
\begin{proof}
The proof follows from Lemma~\ref{lem:vanillaSGDconv} with the exact same steps as in the proof of Theorem~\ref{thm:SGDAIS}.
\end{proof}
{Finally, it is also straightforward to adapt the bias result in Theorem~\ref{thm:SGDAISbias} to this case, which results in the similar bound. We skip it for space reasons and also because it has the same form as in Theorem~\ref{thm:SGDAISbias} with an extra $\log t$ factor.}

\section{Conclusions}\label{sec:conc}
We have presented and analysed \textit{optimised} parametric adaptive importance samplers and provided non-asymptotic convergence bounds for the MSE of these samplers. Our results display the precise interplay between the number of iterations $t$ and the number of samples $N$. In particular, we have shown that the optimised samplers converge to an optimal proposal as $t\to\infty$, leading to an asymptotic rate of $\mathcal{O}(\sqrt{\rho^\star/N})$. This intuitively shows that the number of samples $N$ should be set in proportion to the minimum $\chi^2$-divergence between the target and the exponential family proposal, as we have shown that the adaptation (in the sense of minimising $\chi^2$-divergence or, equivalently, the variance of the weight function) cannot improve the error rate beyond $\mathcal{O}(\sqrt{\rho^\star/N})$. The error rates in this regime may be dominated by how close the target is to the exponential family.

Note that the algorithms we have analysed require constant computational load at each iteration and the computational load does not increase with $t$ as we do not re-use the samples in past iterations. Such schemes, however, can also be considered and analysed in the same manner. More specifically, in the present setup the computational cost of each iteration depends on the cost of evaluating $\Pi(x)$.

Our work opens up several other paths for research. One direction is to analyse the methods with more advanced optimisation algorithms. Another challenging direction is to consider more general proposals than the natural exponential family, which may lead to non-convex optimisation problems of adaptation. Analysing and providing guarantees for this general case would provide foundational insights for general adaptive importance sampling procedures. Also, as shown by \citet{ryu2016convex}, similar theorems can also be proved for $\alpha$-divergences.

Another related piece of work arises from variational inference \citep{wainwright2008graphical}. In particular, \citet{dieng2017variational} have recently considered performing variational inference by minimising the $\chi^2$-divergence, which is close to the setting in this paper. In particular, the variational approximation of the target distribution in the variational setting coincides with the proposal distribution we consider within the importance sampling context in this paper. This also implies that our results may be used to obtain finite-time guarantees for the expectations estimated using the variational approximations of target distributions.

Finally, the adaptation procedure can be modified to handle the non-convex case as well. In particular, the SGD step can be converted into a stochastic gradient Langevin dynamics (SGLD) step. The SGLD method can be used as a global optimiser when $\rho$ and $R$ are non-convex and a global convergence rate can be obtained using the standard SGLD results, see, e.g., \citet{raginsky2017non,zhang2019nonasymptotic}. Global convergence results for other adaptation schemes such as stochastic gradient Hamiltonian Monte Carlo (SGHMC) can also be achieved using results from nonconvex optimisation literature, see, e.g., \citet{akyildiz2020nonasymptotic}.
%
\section*{Acknowledgements}
\"O.~D.~A. is funded by the Lloyds Register Foundation programme on Data Centric Engineering through the London Air Quality project. This work was supported by The Alan Turing Institute for Data Science and AI under EPSRC grant EP/N510129/1. J.~M. acknowledges the support of the Spanish \textit{Agencia Estatal de Investigaci\'on} (awards TEC2015-69868-C2-1-R ADVENTURE and RTI2018-099655-B-I00 CLARA) and the Office of Naval Research (award no. N00014-19-1-2226).




%
\clearpage
\appendix
\section{Appendix}

%
\subsection{Proof of Theorem~\ref{thm:ISfund}}\label{app:proofIS} We first note the following inequalities,
\begin{align*}
|(\varphi,\pi) - (\varphi,\pi_\theta^N)| &= \left| \frac{(\varphi W_\theta, q_\theta)}{(W_\theta,q_\theta)} - \frac{(\varphi W_\theta, q_\theta^N)}{(W_\theta,q_\theta^N)} \right| \\
&\leq \frac{\left|(\varphi W_\theta, q_\theta) - (\varphi W_\theta, {q}_\theta^N)\right|}{|(W_\theta,q_\theta)|} \\ &+ |(\varphi W_\theta, q_\theta^N)| \left| \frac{1}{(W_\theta,q_\theta)} - \frac{1}{(W_\theta,q_\theta^N)} \right| \\
&= \frac{\left|(\varphi W_\theta, q_\theta) - (\varphi W_\theta, q_\theta^N)\right|}{|(W_\theta,q_\theta)|} \\ &+ \|\varphi\|_\infty {|(W_\theta, q_\theta^N)|} \left| \frac{(W_\theta,q_\theta^N) - (W_\theta,q_\theta)}{(W_\theta,q_\theta){(W_\theta,q_\theta^N)}} \right| \\
&= \frac{\left|(\varphi W_\theta, q_\theta) - (\varphi W_\theta, q_\theta^N)\right|}{(W_\theta,q_\theta)} \\&+ \frac{\|\varphi\|_\infty |(W_\theta,q_\theta^N) - (W_\theta,q_\theta)|}{(W_\theta,q_\theta)}.
\end{align*}
We take squares of both sides and apply the inequality $(a+b)^2 \leq 2(a^2 + b^2)$ to further bound the rhs,
\begin{align*}
|(\varphi,\pi) - (\varphi,{\pi}_\theta^N)|^2 &\leq 2 \frac{\left|(\varphi W_\theta, q_\theta) - (\varphi W_\theta, q_\theta^N)\right|^2}{(W_\theta,q_\theta)^2} \\ &+ 2 \frac{\|\varphi\|_\infty^2 |(W_\theta,q_\theta^N) - (W_\theta,q_\theta)|^2}{(W_\theta,q_\theta)^2}
\end{align*}
We now take the expectation of both sides,
\begin{align*}
\bE\left[\left((\varphi,\pi) - (\varphi,{\pi}_\theta^N)\right)^2\right] \leq & \frac{2 \bE\left[\left((\varphi W_\theta, q_\theta) - (\varphi W_\theta, q_\theta^N)\right)^2\right]}{(W_\theta,q_\theta)^2} +
\\
&\frac{2 \|\varphi\|_\infty^2 \bE\left[\left((W_\theta,q_\theta^N) - (W_\theta,q_\theta)\right)^2\right]}{(W_\theta,q_\theta)^2}.
\end{align*}
Note that, both terms in the right hand side are perfect Monte Carlo estimates of the integrals. Bounding the MSE of these integrals yields
\begin{align*}
\bE\left[\left((\varphi,\pi) - (\varphi,{\pi}_\theta^N)\right)^2\right] &\leq \frac{2}{N} \frac{(\varphi^2 W_\theta^2,q_\theta) - (\varphi W_\theta,q_\theta)^2}{(W_\theta,q_\theta)^2} + \frac{2\|\varphi\|_\infty^2}{N} \frac{(W_\theta^2,q_\theta) - (W_\theta,q_\theta)^2}{(W_\theta,q_\theta)^2}, \\
&\leq \frac{2 \|\varphi\|_\infty^2}{N} \frac{(W_\theta^2,q_\theta)}{(W_\theta,q_\theta)^2} + \frac{2\|\varphi\|_\infty^2}{N} \frac{(W_\theta^2,q_\theta) - (W_\theta,q_\theta)^2}{(W_\theta,q_\theta)^2}.
\end{align*}
Therefore, we can straightforwardly write,
\begin{align*}
\bE\left[\left((\varphi,\pi) - (\varphi,{\pi}_\theta^N)\right)^2\right] \leq & \frac{4 \|\varphi\|_\infty^2}{(W_\theta,q_\theta)^2} \frac{(W_\theta^2,q_\theta)}{N}.
\end{align*}
Now it remains to show the relation of the bound to $\chi^2$ divergence. Note that,
\begin{align*}
\frac{(W_\theta^2,q_\theta)}{(W_\theta,q_\theta)^2} &= \frac{\int \frac{\Pi^2(x)}{q_\theta^2(x)} q_\theta(x) \mbox{d}x}{\left(\int \frac{\Pi(x)}{q_\theta(x)} q_\theta(x) \mbox{d}x\right)^2}\\
&= \frac{Z^2 \int \frac{\pi^2(x)}{q_\theta^2(x)} q_\theta(x) \mbox{d}x}{Z^2 \left(\int \pi \mbox{d}x\right)^2}\\
&= \bE_{q_\theta}\left[\frac{\pi^2(X)}{q_\theta^2(X)}\right] := \rho(\theta).
\end{align*}
Note that $\rho$ is not exactly $\chi^2$ divergence, which is defined as $\rho - 1$. Plugging everything into our bound, we have the result,
\begin{align*}
\bE\left[\left((\varphi,\pi) - (\varphi,{\pi}_\theta^N)\right)^2\right] \leq & \frac{4 \|\varphi\|_\infty^2 \rho(\theta)}{N}.
\end{align*}
$\square$
\subsection{Proof of Lemma~\ref{prop:rhoconvex}} \label{app:proofLemma1} We adapt this proof from \citet{ryu2014adaptive} by following the same steps. We first show that $A(\theta)$ is convex by first showing that $\exp(A(\theta))$ is convex. Choose $0 < \eta < 1$ and using H\"{o}lder's inequality,
\begin{align*}
\exp(A(\eta\theta_1 &+ (1-\eta)\theta_2)) = \int \exp((\eta \theta_1 + (1-\eta) \theta_2)^\top T(x)) h(x) \mbox{d}x \\
&= \int \left(\exp(\theta_1^\top T(x))h(x)\right)^\eta \left(\exp(\theta_2^\top T(x)) h(x)\right)^{1-\eta} \mbox{d}x \\
&\leq \left( \int \exp(\theta_1^\top T(x))h(x)\mbox{d}x\right)^\eta \left(\int \exp(\theta_2^\top T(x))h(x)\mbox{d}x\right)^{1-\eta}.
\end{align*}
Taking $\log$ of both sides yields
\begin{align*}
A(\eta \theta_1 + (1-\eta) \theta_2) \leq \eta A(\theta_1) + (1-\eta) A(\theta_2),
\end{align*}
which shows the convexity of $A(\theta)$. Note that $A(\theta) - \theta^\top T(x)$ is convex in $\theta$ since it is a sum of a convex and a linear function of $\theta$. Since $\exp$ is an increasing convex function and the composition of convex functions is convex, $M(\theta,x) := \exp(A(\theta) - \theta^\top T(x))$ is convex in $\theta$. Finally we prove that $\rho(\theta)$ is convex. First let us write it as
\begin{align*}
\rho(\theta) &= \int \frac{\pi^2(x)}{q_\theta(x)^{{2}}} {q_\theta(x)} \mbox{d}x = \int \frac{\pi^2(x)}{h(x)} M(\theta,x) \mbox{d}x.
\end{align*}
Then we have the following sequence of inequalities
\begin{align*}
\rho(\eta \theta_1 + (1-\eta) \theta_2) &= \int \frac{\pi^2(x)}{h(x)} M(\eta \theta_1 + (1-\eta) \theta_2,x) \mbox{d}x \\
&\leq \int \frac{\pi^2(x)}{h(x)} (\eta M(\theta_1,x) + (1-\eta) M(\theta_2,x)) \mbox{d}x\\
&= \eta \int \frac{\pi^2(x)}{h(x)} M(\theta_1,x) \mbox{d}x + (1-\eta) \int \frac{\pi^2(x)}{h(x)} M(\theta_2,x) \mbox{d}x\\
&= \eta \rho(\theta_1) + (1-\eta) \rho(\theta_2),
\end{align*}
which concludes the claim. $\square$

\subsection{Proof of Theorem~\ref{thm:GD}}\label{app:thm:GD} First note that, using Theorem~\ref{thm:ISfundAIS}, we have
\begin{align*}
\bE\left[\left((\varphi,\pi) - (\varphi,\pi_{\theta_t}^N)\right)^2\right] &\leq \frac{c_\varphi \rho(\theta_t)}{{N}},\\
&= \frac{c_\varphi (\rho(\theta_t) - \rho^\star)}{{N}} + \frac{c_\varphi \rho^\star}{{N}}, \\
&\leq \frac{c_\varphi \|\theta_0 - \theta^\star\|^2}{2 \gamma t {N}} +  \frac{c_\varphi \rho^\star}{{N}},
\end{align*}
where the last inequality follows from Lemma~\ref{lem:GDconv}.
$\square$

\subsection{A sufficient condition for Assumption \ref{ass:SupSupBound} to hold} \label{apRho2}

{
Recall that we have defined $\rho_2(\theta) = \bE[w_\theta^4(X) ] = \bE\left[ \frac{\pi^4(X)}{q_\theta^4(X)}\right]$ and $q_\theta(x) = \exp\left\{ \left(\theta^\top T(x) - A(\theta)\right)h(x) \right\}$ whenever $q_\theta(x)$ belongs to the exponential family. We have the following result.
}

{
\begin{prop}
Let the $\rho_2$ be Lipschitz with Lipschitz derivatives, let $q_\theta(x)$ belong to the exponential family and let $\Theta$ be compact. If the sufficient statistics $T(X)$ of the distribution $q_\theta$ all have finite variances, i.e., 
$$
\max_{i=1,...,d_\theta} \text{Var}(T_i(X)) < \infty,
$$
then Assumption \ref{ass:SupSupBound} holds.
\end{prop}
\noindent \textit{Proof.} Using the fact that
$$
\frac{
	\partial q_\theta(x)
}{
	\partial \theta_i
} = q_\theta(x) \frac{
	\partial \log q_\theta(x)
}{
	\partial \theta_i
} 
$$
one can readily calculate the second order derivatives of $\rho_2(\theta)$. In particular,
\begin{eqnarray}
\frac{
	\partial^2 \rho_2(\theta)
}{
	\partial \theta_i^2 
} &=& 9\bE\left[ 
	w_\theta^4(X) \left(
		T_i(X) - \frac{
			\partial A(\theta)
		}{
			\partial \theta_i
		}
	\right)^2
\right] \nonumber\\
&& + 3\bE\left[
	w_\theta^4(X)
\right] \frac{
	\partial^2 A(\theta)
}{
	\partial \theta_i^2
} < \infty,
\end{eqnarray}
where the inequality holds because $\rho_2(\theta)$ has Lipschitz derivatives in $\Theta$. However, $\frac{
	\partial^2 A(\theta)
}{
	\partial \theta_i^2
} = \text{Var}(T_i(X))$ and, by assumption, $\max_i \text{Var}(T_i(X))<\infty$. Moreover, $\bE\left[
	w_\theta^4(X)
\right] = \rho_2(\theta) < \infty$ because $\rho_2(\theta)$ is Lipschitz and the parameter space $\Theta$ is compact. Therefore, it follows that 
$$
\bE\left[ 
	w_\theta^4(X) \left(
		T_i(X) - \frac{
			\partial A(\theta)
		}{
			\partial \theta_i
		}
	\right)^2
\right] < \infty
$$
and Assumption \ref{ass:SupSupBound} holds. $\square$
}

{\subsection{Proof of Lemma~\ref{lem:GradientMonteCarlo}}\label{app:lem:GradientMonteCarlo}
We first note that the exact gradient can be written as
\begin{align*}
\nabla_\theta \bE_{q_\theta}\left[\frac{\pi^2(X)}{q_\theta^2(X)}\right] &= \nabla_\theta \int \frac{\pi^2(x)}{q_\theta(x)} \mbox{d}x \\
&= - \int \frac{\pi^2(x)}{q_\theta^2(x)} \nabla_\theta \log q_\theta(x) q_\theta(x) \mbox{d}x.
\end{align*}
Now, note that,
$$
\nabla_\theta \log q_\theta(x) = \left[
	\begin{array}{c}
	\frac{\partial \log q_\theta}{\partial \theta_{1}}\\
	\frac{\partial \log q_\theta}{\partial \theta_{2}}\\
	\vdots\\
	\frac{\partial \log q_\theta}{\partial \theta_{d_\theta}}\\
	\end{array}
\right].
$$
Given the samples $x_t^{(i)} \sim q_{\theta_{t-1}}$ for $i = 1,\ldots,N$ to estimate the gradient, we can write the mean-squared error $\bE\|g_t - \nabla\rho(\theta_{t-1})\|_2^2$ as
\begin{align*}
\bE&\left[\left\|\frac{1}{N} \sum_{i=1}^N \frac{\pi^2(x_t^{(i)})}{q_{\theta_{t-1}}^2(x_t^{(i)})} \nabla_\theta \right.\right.\left.\left. \log q_{\theta_{t-1}}(x_t^{(i)}) -  \int \frac{\pi^2(x)}{q_\theta^2(x)} \nabla_\theta \log q_{\theta_{t-1}}(x) q_{\theta_{t-1}}(x) \mbox{d}x \right\|_2^2\right] \\&= \sum_{j=1}^{d_\theta} \bE\left[\left(\frac{1}{N}\sum_{i=1}^N \frac{\pi^2(x_t^{(i)})}{q_{\theta_{t-1}}^2(x_t^{(i)})} \frac{\partial \log q_{\theta_{t-1}}}{\partial \theta_j}(x_t^{(i)}) - \int \frac{\pi^2(x)}{q_{\theta_{t-1}}^2(x)} \frac{\partial \log q_{\theta_{t-1}}}{\partial \theta_j}(x) q_{\theta_{t-1}}(x) \mbox{d}x \right)^2\right].
\end{align*}
Now the expectation is a standard Monte Carlo error for the test function,
\begin{align*}
\varphi_j(x) = \frac{\pi^2(x)}{q_\theta^2(x)} \frac{\partial \log q_\theta}{\partial \theta_j}(x).
\end{align*}
Assumption~\ref{ass:SupSupBound} together with Lemma~A.1 in \citet{crisan2014particle} yields
\begin{align*}
\bE[\| g_t - \nabla \rho(\theta)\|_2^2] \leq \frac{d_\theta c_\rho D_\pi}{N},
\end{align*}
where $c_\rho < \infty$ and $D_\pi = \max_{j \in \{1,\ldots,d_\theta\}} D_\pi^j<\infty$ are constants independent of $N$. $\square$}
\subsection{Proof of Lemma~\ref{lem:SGDconv}}\label{app:lem:SGDconv} Since projections reduce distances, we have,
\begin{align*}
\|\theta_{k} - \theta^\star\|_2^2 &\leq \|\theta_{k-1} - \gamma_{k} g_{k} - \theta^\star\|_2^2 \\
&= \|\theta_{k-1} - \theta^\star\|_2^2 - 2 \gamma_k g_{k}^\top (\theta_{k-1} - \theta^\star) + \gamma^2_{k} \|g_{k}\|_2^2.
\end{align*}
Let $\cF_{k-1} = \sigma(\theta_0,\ldots,\theta_{k-1},g_1,\ldots,g_{k-1})$ be the $\sigma$-algebra generated by random variables $\theta_0,\ldots,\theta_{k-1}$, $g_1,\ldots,g_{k-1}$ and take the conditional expectations with respect to $\cF_{k-1}$
\begin{align*}
\bE\left[\|\theta_{k} - \theta^\star\|_2^2 | \cF_{k-1}\right] &\leq \|\theta_{k-1} - \theta^\star\|_2^2 - 2 \gamma_k \nabla \rho(\theta_{k-1})^\top (\theta_{k-1} - \theta^\star) + \gamma_{k}^2 \bE\left[\|g_{k}\|_2^2 | \cF_{k-1}\right].
\end{align*}
Next, using the convexity of $\rho$ yields
\begin{align*}
\bE\left[\|\theta_{k} - \theta^\star\|_2^2 | \cF_{k-1}\right] &\leq \|\theta_{k-1} - \theta^\star\|_2^2 - 2 \gamma_k [\rho(\theta_{k-1}) - \rho(\theta^\star)] + \gamma^2_{k} \bE\left[\|g_{k}\|_2^2 | \cF_{k-1}\right].
\end{align*}
Finally, we take unconditional expectations of both sides,
\begin{align*}
\bE\|\theta_{k} - \theta^\star\|_2^2\leq 
\bE\|\theta_{k-1} - \theta^\star\|_2^2 - 2 \gamma_k \bE[(\rho(\theta_{k-1}) - \rho(\theta^\star)]+ \gamma^2_{k} \bE\|g_{k}\|_2^2.
\end{align*}
With rearranging, using $\bE\|g_k - \nabla\rho(\theta_{k-1})\|_2^2 = \bE\|g_k\|^2 - \|\nabla\rho(\theta_{k-1})\|^2$ and invoking Assumption~\ref{ass:BoundedGradient}, we arrive at
\begin{align*}
\bE[\rho(\theta_{k-1}) - \rho(\theta^\star)] \leq \frac{\bE\|\theta_{k-1} - \theta^\star\|_2^2 - \bE\|\theta_{k} - \theta^\star\|_2^2}{2\gamma_k} + \frac{\gamma_{k} (\sigma_\rho^2 + G^2_\rho)}{2}.
\end{align*}
where {$\sigma_\rho^2 = d_\theta D_\pi / N$ as given in Lemma~\ref{lem:GradientMonteCarlo}}. Now summing both sides from $k = 1$ to $t$ and dividing both sides by $t$,
\begin{align*}
\bE[\rho(\bar{\theta}_t) - \rho(\theta^\star)] &\leq \frac{1}{t} \sum_{k=1}^t \bE[\rho(\theta_{k-1}) - \rho(\theta^\star)] \leq \frac{\bE\|\theta_0 - \theta^\star\|_2^2}{2\gamma_t t} + \sum_{k=1}^t \frac{\gamma_k (\sigma_\rho^2 + G^2_\rho)}{2t},
\end{align*}
since $\frac{1}{\gamma_k} \leq \frac{1}{\gamma_t}$ for all $k\leq t$. Substituting $\gamma_k = \alpha/\sqrt{k}$ and noting that
\begin{align*}
\sum_{k=1}^t \frac{1}{\sqrt{k}} \leq \int_0^t \frac{1}{\sqrt{\tau}}\md \tau = 2 \sqrt{t},
\end{align*}
we arrive at
\begin{align*}
\bE[\rho(\bar{\theta}_t) - \rho(\theta^\star)] \leq \frac{\bE\|\theta_0 - \theta^\star\|_2^2}{2\alpha \sqrt{t}} + \frac{\alpha (\sigma_\rho^2 + G^2_\rho)}{\sqrt{t}},
\end{align*}
where $\bar{\theta}_t = \frac{1}{t} \sum_{k=0}^{t-1} \theta_k$. $\square$
\subsection{Proof of Theorem~\ref{thm:SGDAIS}}\label{app:thm:SGDAIS} Let $\cF_{t-1} = \sigma(\theta_0,\ldots,\theta_{t-1},g_1,\ldots,g_{t-1})$ be the $\sigma$-algebra generated by the random variables $\theta_0,\ldots,\theta_{t-1},g_1,\ldots,g_{t-1}$. Then
\begin{align*}
\bE\left.\left[\left((\varphi,\pi) - (\varphi,\pi^N_{\bar{\theta}_t})\right)^2  \right\vert \cF_{t-1} \right] &\leq \frac{c_\varphi \rho (\bar{\theta}_t)}{{N}}\\
&= \frac{c_\varphi (\rho(\bar{\theta}_t) - \rho^\star)}{{N}} + \frac{c_\varphi \rho^\star}{{N}},
\end{align*}
where $\bar{\theta}_t = \frac{1}{t} \sum_{k=0}^{t-1} \theta_k$ is an $\cF_{t-1}$-measurable random variable. Now if we take unconditional expectations of both sides,
\begin{align*}
\bE\left[\left((\varphi,\pi) - (\varphi,\pi^N_{\bar{\theta}_t})\right)^2\right] &\leq \frac{c_\varphi \bE\left[(\rho(\bar{\theta}_t) - \rho^\star)\right]}{{N}} + \frac{c_\varphi \rho^\star}{{N}}.
\end{align*}
The result follows from applying Lemma~\ref{lem:SGDconv} for $\bE\left[(\rho(\bar{\theta}_t) - \rho^\star)\right]$.
$\square$

\bibliography{../STCO_v2/draft}
\clearpage

\end{document}